\documentclass{article}
\usepackage[utf8]{inputenc}

\usepackage{epsfig}
\usepackage{graphicx}
\usepackage{amsbsy}
\usepackage{amsmath}
\usepackage{amsfonts}
\usepackage{amssymb}
\usepackage{textcomp}
\usepackage{hyperref}
\usepackage{aliascnt}
\usepackage{mathrsfs}

\newcommand{\mcm}[3]{\newcommand{#1}[#2]{{\ensuremath{#3}}}} 

\mcm{\tuple}{1}{\langle #1 \rangle}
\mcm{\name}{1}{\ulcorner #1 \urcorner}
\mcm{\Nbb}{0}{\mathbb{N}}
\mcm{\Zbb}{0}{\mathbb{Z}}
\mcm{\Rbb}{0}{\mathbb{R}}
\mcm{\Cbb}{0}{\mathbb{C}}
\mcm{\Qbb}{0}{\mathbb{Q}}
\mcm{\Acal}{0}{\cal A}
\mcm{\Bcal}{0}{\cal B}
\mcm{\Ccal}{0}{\cal C}
\mcm{\Dcal}{0}{\cal D}
\mcm{\Ecal}{0}{\cal E}
\mcm{\Fcal}{0}{\cal F}
\mcm{\Gcal}{0}{\cal G}
\mcm{\Hcal}{0}{\cal H}
\mcm{\Ical}{0}{\cal I}
\mcm{\Jcal}{0}{\cal J}
\mcm{\Kcal}{0}{\cal K}
\mcm{\Lcal}{0}{\cal L}
\mcm{\Mcal}{0}{\cal M}
\mcm{\Ncal}{0}{\cal N}
\mcm{\Ocal}{0}{{\cal O}}
\mcm{\Pcal}{0}{{\cal P}}
\mcm{\Qcal}{0}{{\cal Q}}
\mcm{\Rcal}{0}{{\cal R}}
\mcm{\Scal}{0}{{\cal S}}
\mcm{\Tcal}{0}{{\cal T}}
\mcm{\Ucal}{0}{{\cal U}}
\mcm{\Vcal}{0}{{\cal V}}
\mcm{\Wcal}{0}{{\cal W}}
\mcm{\Xcal}{0}{{\cal X}}
\mcm{\Ycal}{0}{{\cal Y}}
\mcm{\Zcal}{0}{{\cal Z}}
\mcm{\Mfrak}{0}{\mathfrak M}

\mcm{\restric}{0}{\upharpoonright}
\mcm{\upset}{0}{\uparrow}
\mcm{\onto}{0}{\twoheadrightarrow}
\mcm{\smallNbb}{0}{{\small \mathbb{N}}}
\DeclareMathOperator{\preop}{op}
\mcm{\op}{0}{^{\preop}}

\newcommand{\se}{\subseteq}
{\begin{array}{c}
\setlength{\unitlength}{1em}}%
{\end{array}}

\usepackage{amsthm}

\newcommand{\theoremize}[2]{\newaliascnt{#1}{thm} \newtheorem{#1}[#1]{#2} \aliascntresetthe{#1}}

\theoremstyle{plain}
\newtheorem{thm}{Theorem}[section]
\theoremize{lem}{Lemma}
\theoremize{skolem}{Skolem}
\theoremize{fact}{Fact}
\newtheorem{sublem}{Claim}[thm]
\theoremize{obs}{Observation}
\theoremize{prop}{Proposition}
\theoremize{cor}{Corollary}
\theoremize{que}{Question}
\theoremize{oque}{Open Question}
\theoremize{oprob}{Open Problem}
\theoremize{con}{Conjecture}
\theoremize{setting}{Setting}

\theoremstyle{definition}
\theoremize{dfn}{Definition}

\theoremize{rem}{Remark}
\theoremize{eg}{Example}
\theoremize{exercise}{Exercise}
\theoremstyle{plain}

\usepackage{graphicx, hyperref, xcolor, nicefrac}
\usepackage{subcaption}
\usepackage{csquotes,amsmath}
\usepackage[british]{babel}
\usepackage{tabularx}
\usepackage{algorithm}
\usepackage{algorithmic}

\newcommand{\algorithmicfunction}{\textbf{function}}
\newcommand{\algorithmicendfunction}{\algorithmicend\ \algorithmicfunction}
\makeatletter
\newcommand\FUNCTION[3][default]{%
\ALC@it \algorithmicfunction\ \textsc{#2}(#3)%
 \ALC@com{#1}%
\begin{ALC@prc}%
}
\newcommand\ENDFUNCTION{%
  \end{ALC@prc}%
  \ifthenelse{\boolean{ALC@noend}}{}{%
    \ALC@it\algorithmicendfunction
  }%
}
\newenvironment{ALC@prc}{\begin{ALC@g}}{\end{ALC@g}}
\makeatother

\title{How to apply tree decomposition ideas \\ in large networks?}
\author{Johannes Carmesin \and Sarah Frenkel}
 \date{\today}

\begin{document}

\maketitle

\begin{abstract}
Graph decompositions are the natural generalisation of tree decompositions where the decomposition tree is replaced by a genuine graph.
Recently they found theoretical applications in the theory of sparsity, topological graph theory, structural graph theory and geometric group theory. 

We demonstrate applicability of graph decompositions on large networks by implementing an efficient algorithm and testing it on road networks. 
\end{abstract}

\section{Introduction}

Tree decompositions are a powerful tool in algorithmic and structural graph theory. Their significance became apparent during the graph minors project~\cite{GMX}. 
By Courcelle's theorem~\cite{{courcelle1990monadic},{courcelle2012graph}} monadic second order graph properties can be verified in linear time on graphs of bounded tree-width. For algorithmic results on computing tree-width see Bodlaender's survey~\cite{bodlaender1998partial}, and for applications of tree decompositions beyond graph theory see \cite{diestel2016tangles} by Diestel and Whittle. 
According to google scholar 19,000 papers have \lq tree decompositions\rq\ in the title. With applications to large networks in mind, Adcock, Sullivan and Mahoney~\cite{adcock2013tree} asked for methods to study graphs that  have bounded tree-width locally but not necessarily globally. Graph decompositions\footnote{Unfortunately, the term \lq graph decomposition\rq\ has several non-related meanings in the literature. A very popular seemingly not directly related one appears in extremal graph theory.} -- the natural generalisation of tree decompositions where the decomposition tree is replaced by a genuine graph -- were suggested as a possible solution~\cite{carmesin2022local}. In this paper we implement an efficient algorithm that computes such graph decompositions; and we demonstrate the applicability on real world networks. 

\vspace{0.3 cm}

The idea of generalising tree decompositions to graph decompositions is due to Diestel and Kuhn from 2005~\cite{diestel2005graph}. At that time it was clear that graph decompositions will eventually become useful as they allow for descriptions of complicated and exciting graph classes in simple terms: just take any small graph $S$ and glue copies of $S$ together along a large decomposition graph $L$. 
However, until recently we had no methods to \emph{compute} graph decompositions of a given graph, not even theoretical ones. In a nutshell the idea that changed this~\cite{carmesin2022local} is the following. 
Tree decompositions can be understood as a recipe how to cut a graph along a set of separators simultaneously. When generalising \lq separators\rq\ to \lq local separators\rq -- meaning, vertex sets that need not separate the graph globally but just locally -- the global structure need no longer be that of a decomposition tree but can be that of a genuine graph. This implies that graph decompositions\footnote{See \autoref{sec3} and \autoref{sec5} for formal definitions.} can be used to study \emph{local-global structure} in large networks via the interplay between local separators and the global decomposition graph. 

Despite this obvious potential for applications on large networks, so far works on local separators and graph decompositions have been of more theoretical nature. 
These include the first explicit characterisation of a class of graphs by forbidden shallow minors in the context of the rapidly growing theory of sparsity, 
for an introduction see the book of 
Ne{\v{s}}et{\v{r}}il and Ossana de 
Mendez~\cite{nevsetvril2012sparsity}. In topological graph theory they allow for a new polynomial algorithm to compute surface embeddings and a corresponding Whitney-type theorem~\cite{whitney_surface}. In structural graph theory they can be used to characterise graphs without cycles of intermediate length~\cite{RajeshStructure}. In geometric group theory, they are used to study finite nilpotent groups and provide a conjecture for a possible extension of Stallings' theorem to detect product structure in finite groups~\cite{StallingsType}. How methods from infinite graph theory can be applied in this new theory is best understood through~\cite{diestel2022canonical}.

Whilst the idea for local separators is based on covering theory from topology, in~\cite{carmesin2022local} a finitary equivalent definition was provided. 
This definition is implementable in principle, however, it turns out to be too complicated to yield practical algorithms. Here, we solve this problem by providing a much simpler equivalent definition, see \autoref{equi-def} below. This theorem is also of theoretical interest; for example in~\cite{StallingsType} it is applied to simplify otherwise fairly technical arguments, saving 10 pages of computation.

The main result of this paper is an efficient algorithm that computes graph decompositions via local separators. We successfully test its applicability by computing local-global structure of large road networks.

A well established way to test an algorithm is to run it on examples with a ground truth and compare the results.
For this we choose to analyse a road network. Road maps do not come with a hard ground truth. But just looking at them we humans are able to identify structure -- like towns, rural areas and main connecting routes between them, see \autoref{fig1} and \autoref{sec7} for details.

\begin{figure}[!ht]
    \centering
     \begin{subfigure}[t]{0.48\textwidth}
     \centering        \includegraphics[scale=0.132]{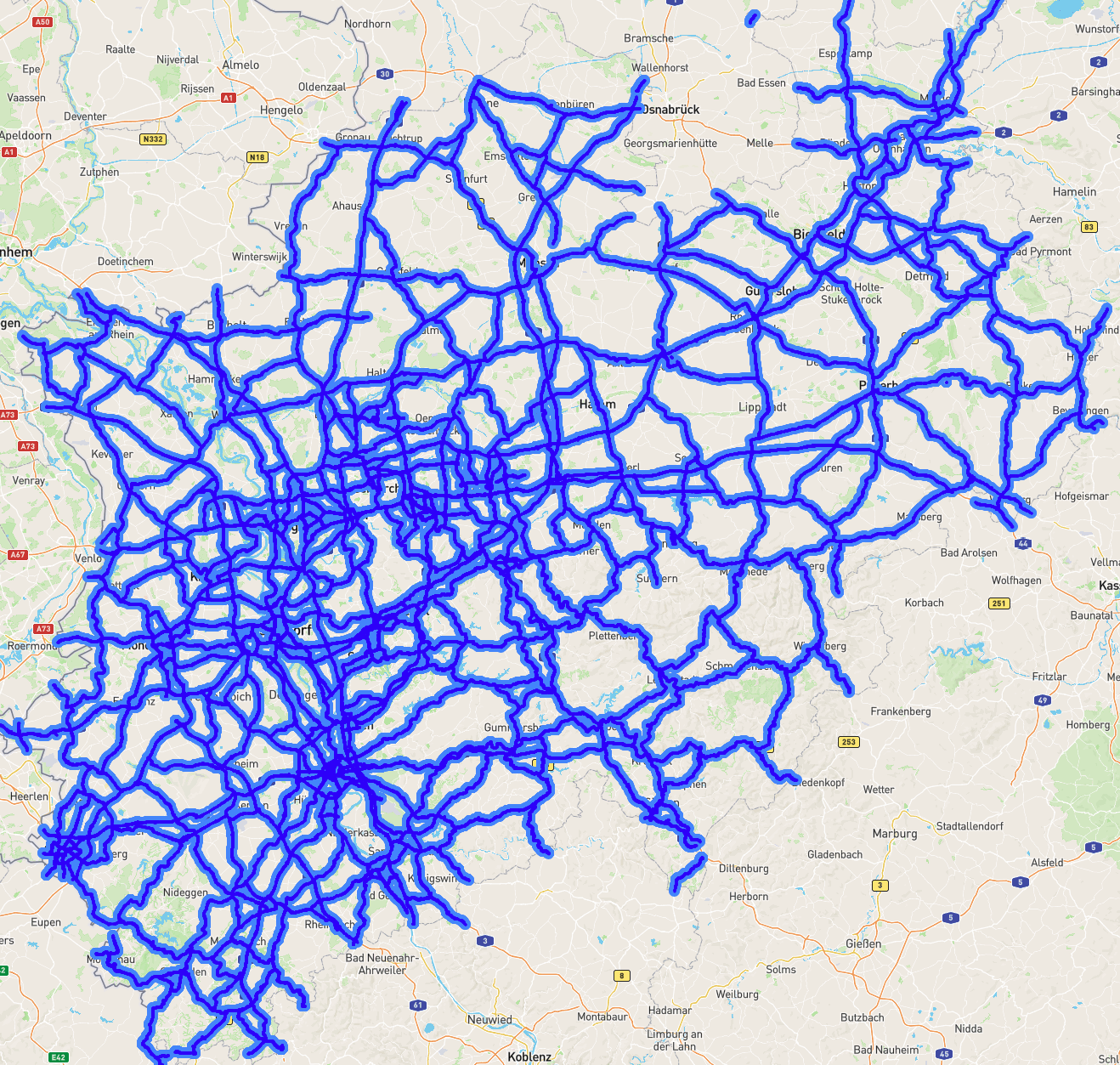} 
    \caption{}
        \label{fig:nrw_raw}
\end{subfigure}
\hfill
    \centering
        \begin{subfigure}[t]{0.48\textwidth}
        \centering
        \includegraphics[scale=0.3]{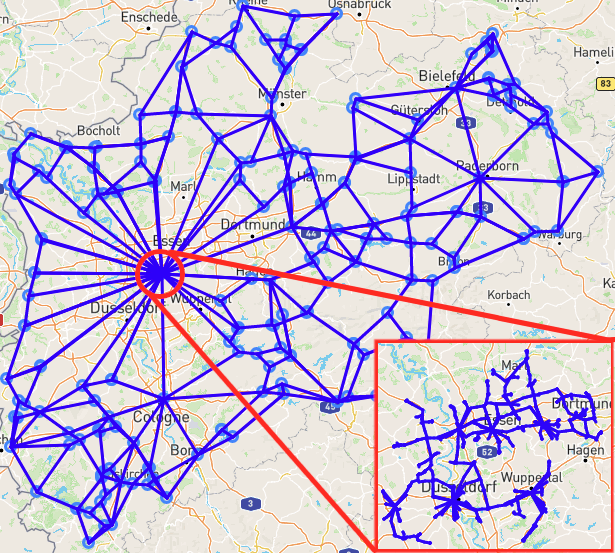} 
   \caption{}        \label{fig:decomp}
        \end{subfigure}
        \caption{North Rhine-Westphalia is a densely populated region in Germany. Its road network has about 300,000 vertices and is depicted on the left. 
         Applying our algorithm to this network, we obtain the figure on the right. The rough structure of North Rhine-Westphalia is compassed by a decomposition graph of 224 nodes, and each node of this decomposition graph represents a cluster. The red rectangle contains a graph that depicts the fine structure at one of the major clusters.}\label{fig1}
\end{figure}
Theoretical estimations of the running time are given in \autoref{sec:algorithms}. A highlight is the fast parallel running time of our algorithms.  The reason being that the decomposition we compute is \emph{canonical} and thus the decomposition can be computed in parallel at various regions of the graph, while the general theory guarantees that such local solutions will combine unambiguously into a unifying global structure. 
This fast parallel running time is particularly useful in the context of large networks.
We have demonstrated that our algorithms runs fast and stable on large networks. Our implementation calculates $1$-separators on a network with about 316,000 vertices and 322,000 edges on a standard computer in under 10 seconds. The next step is to run our algorithm on networks without ground truth to compute new local-global structure.

\vspace{.3cm}
{\bf Related Methods.} 
Constructing graph decompositions via local separators is a new tool that computes the structure of large networks. 
Next we summarise well established tools for computing structure in large networks; for a detailed account we refer to the surveys~\cite{fortunato2010community, javed2018community}.

The problem is to decompose a graph into communities (or clusters). These are vertex-sets that have only few ties with the rest of the network. 
The different approaches provide different interpretations of what constitutes \enquote{few ties}.
When the number $k$ of clusters is known advance, one can apply the $k$-means clustering~\cite{macqueen1967some}, where a graph is partitioned into $k$ vertex sets via optimization methods.
The Girvan-Newman algorithm~\cite{newman2004finding} deletes edges from a network with respect to some \enquote{betweenness measure} to derive a partition of the vertex set into clusters. Since this measure is recomputed after each deletion, this algorithm does not scale to large graphs easily. 
Random walks are a probabilistic tool that can be used to compute communities~\cite{rosvall2008maps}. 
In spectral graph theory the eigenvalues of the Laplacian of a graph are used to compute clusters~\cite{von2007tutorial}.  Adcock, Mulligan and Sullivan used 
tree decompositions directly to find clusters in graphs~\cite{adcock2013tree}, and put forward the above mentioned challenge.
Klepper et al. used tangles to find clusters in graphs~\cite{klepper2023clustering}.

The remainder of this paper is structured as follows. The short \autoref{sec:original-definitions} explains background material from \cite{carmesin2022local}. We introduce our new perspective on local 2-separators in \autoref{sec3}. In \autoref{sec5} we explain how to obtain a graph decomposition based on local separators.
We present our algorithms to compute local separators in \autoref{sec:algorithms}.  
We then apply these algorithms to analyse a large road network in \autoref{sec7} . We finish in \autoref{sec:conclusion} with our conclusions.

\section{Original definition of local separators}\label{sec:original-definitions}

Local separators were introduced in \cite{carmesin2022local}. 
In \autoref{sec3} we will present an alternative definition, better suited for the algorithmic scope of this paper.
In this section we present the terminology required to prove equivalence between these two definitions. We start with the definition of \emph{$d$-local cutvertices} and then move on to \emph{$d$-local $2$-separators}.
Given an integer $d$ and a graph $G$ with a vertex $v$, \emph{the ball} of diameter $d$ around $v$ is the subgraph of $G$ consisting of those vertices and edges of $G$ on closed walks of length at most $d$ containing $v$; we denote the ball by $D_d(v)$.

\begin{dfn}[{\cite[Definition 3.1.]{carmesin2022local}}]\label{def:Expl}
Now we give a formal definition of the explorer-neighbourhood of parameter $d$ in a graph $G$ with explorers based at the vertices $v$ and $w$ with distance at most $d/2$. The \emph{core} is the set of all vertices on shortest paths between the vertices $v$ and $w$. We take a copy of the ball $D_d(v)$ where we label a vertex $u$ with the set of shortest paths from the core to $u$ contained in the ball $D_d(v)$. Similarly, we take a copy of the ball $D_d(w)$ where we label a vertex $u$ with the set of shortest paths from the core to $u$ contained in the ball $D_d(w)$. Now we take the union of these two labelled balls -- with the convention that two vertices are identified if they have a common label in their sets; that is, there is a shortest path from the core to that vertex discovered by both explorers. (Note that the same vertex $x$ of $G$ could be in both balls but the label sets could be disjoint. In this case there would be two copies of that vertex in the union. In such a case the union would not be a subgraph of the original graph). We denote the explorer neighbourhood by $Ex_d(v, w)$. 
\end{dfn}

\begin{dfn}[{\cite[Definition 3.7]{carmesin2022local}}]\label{def:2sepEx} Given a graph $G$ with distinct vertices $v_0$ and $v_1$, we say that the set $\{v_0,v_1\}$ is a $d$-local $2$-separator if the vertices $v_0$ and $v_1$ have distance at most $d/2$ in the graph $G$ and 
the punctured explorer-neighbourhood $Ex_d(v_0,v_1) - v_0 - v_1$ is disconnected.
\end{dfn}

\begin{lem}[{\cite[Lemma 3.4.]{carmesin2022local}}]\label{lem:3.4.} Let $o$ be a cycle (or more generally a closed walk) of length at most $d$ containing vertices $v_0$ and $v_1$. Vertices of $o$ have unique copies in $Ex_d(v_0,v_1)$. 
\end{lem}

The \emph{edge-space} $\mathcal{E}(G)$ of a graph $G$ is the vector space $\mathbb{F}_2^E$; that is the vector space that has one coordinate for every edge $e\in E$ of $G$ over the field $\mathbb{F}_2$.
The \emph{characteristic vector} of an edge set $F$ of a graph $G$ is the vector of the edge-space that takes the value $1$ at exactly those edges that are in $F$.
The \emph{support} of a vector is the set of its coordinates whose entries are nonzero. 

\begin{eg}The support of a characteristic vector of an edge set $F$ is equal to $F$. 
\end{eg}
We say that a cycle $o$ of a graph $G$ is \emph{generated} from a family $\Ocal$ of cycles if 
 in the edge space of $G$,
the characteristic vector of (the edge set of) $o$ is generated by the family of characteristic vectors of cycles from $\Ocal$.

\smallskip
The balls $D_d(v_0)$ and $D_d(v_1)$ are embedded within the explorer neighbourhood as described in its construction. We refer to these embedded balls as $D_d(v_0)'$ and $D_d(v_1)'$.

\begin{lem}\label{lem:3.6.}
Every closed walk $o$ of the explorer neigh\-bour\-hood $Ex_d(v_0, v_1)$ is generated from the cycles of the embedded balls $D_d(v_0)'$ and $D_d(v_1)'$ of length at most $d$.  
\end{lem}

\begin{proof}
\autoref{lem:3.6.} follows from combining a variant of \cite[Lemma 3.6]{carmesin2022local} with  \cite[Lemma 2.4]{carmesin2022local};
this variant of \cite[Lemma 3.6]{carmesin2022local} is its strengthening with \lq closed walk\rq\ in place of \lq cycle\rq, which is trivially equivalent to 
 \cite[Lemma 3.6]{carmesin2022local}.
\end{proof}

\begin{rem}[{\cite[Remark 3.5.]{carmesin2022local}}]\label{lem:3.5.}
Every cycle of $Ex_d(v_0, v_1)$ of length at most $d$ containing one of the vertices $v_0$ or $v_1$, say $v_0$, is a cycle of $G$. Indeed, it is contained in the ball of diameter $d$ around $v_0$ and as such a cycle of $G$.
\end{rem}

\begin{lem}[{\cite[Lemma 3.10.]{carmesin2022local}}]\label{lem:3.10.} Let $\{v_0, v_1\}$ be a $d$-local $2$-separator in a $d$-locally $2$-connected graph $G$. For every connected component $K$ of the punctured explorer neighbourhood $Ex_d(v_0, v_1) - v_0 - v_1$, there is a cycle $o'$ of length at most $d$ containing the vertices $v_0$ and $v_1$, and $o'$ contains a vertex of the component $K$ and $o'$ contains an edge incident with $v_0$ whose other end vertex is a vertex not in $K$.
\end{lem}

\section{A new perspective on local 2-separators}\label{sec3}

In this section we propose an alternative definition of local 2-separators and show that it is equivalent to the one introduced in \cite{carmesin2022local}.
This definition is simpler and has better algorithmic properties and is the basis of the algorithm presented in this paper. 
Given a vertex set $X$, we denote by $N(X)$ the \emph{set of neighbours of $X$}; that is, the vertices outside $X$ that have a neighbour in $X$. 

\begin{dfn}[Connectivity graph]\label{def:ConnectivityGraph}
Given a graph $G$, an integer $d$ and two distinct vertices $v_0,v_1 \in V(G)$ with distance at most $d/2$,
the \emph{connectivity graph} $C$ has the vertex set $N=N(\{v_0,v_1\})$.
And $xy$ is an edge of $C$ if there is a path between $x$ and $y$ in $D_d(v_i)-v_0-v_1$ for some $i=0,1$. 

We denote the connectivity graph by $C_d(v_0,v_1,G)$ or simply $C$ when the other data is clear from the context.
\end{dfn}

The goal of this section is to prove the following:

\begin{thm}\label{equi-def}
Let $d\in \mathbb{N}$ with $d\geq 2$, and $G$ be a graph with vertices $v_0$ and $v_1$ of distance at most $d/2$.
Then $\{v_0,v_1\}$ is a $d$-local separator if and only if the connectivity graph $C_d(v_0,v_1,G)$ is disconnected.
\end{thm}

\begin{lem}\label{lem:neigbInH}
Let $x$ and $y$ be two adjacent vertices in the connectivity graph $\mathcal C_d(v_0,v_1)$. Then the vertices $x$ and $y$ have unique copies in
$Ex_d(v_0,v_1)$ and these unique copies are in the same component of the punctured explorer neighbourhood $Ex_d(v_0,v_1)-v_0-v_1$.
\end{lem}

\begin{proof}
As $x$ and $y$ are adjacent in $C$, there is some $i\in \{0,1\}$ such that there is an $x$-$y$-path $P$ contained in $D_d(v_i)-v_0-v_1$.
As $x$ and $y$ are vertices of $C$, they are in $N$, so there are vertices $v_x,v_y\in \{v_0,v_1\}$ so that $xv_x$ and $yv_y$ are edges.
Then $v_xPv_y$ is a walk in the explorer neighbourhood $Ex_d(v_0,v_1)$. Now we define a closed walk $o$. If the walk $v_xPv_y$ is closed, we set $o$ to be equal to it; otherwise its end vertices are $v_0$ and $v_1$ and we obtain $o$ from it by adding a shortest $v_0$-$v_1$-path.
This ensures that $o$ is a closed walk of the explorer neighbourhood $Ex_d(v_0,v_1)$ that contains any of the edges $xv_x$ and $yv_y$ exactly once -- unless the vertex $x$ or $y$, respectively, is in the core and thus has a unique copy in $Ex_d(v_0,v_1)$.

According to \autoref{lem:3.6.} $o$ is generated by a set  of cycles $\mathcal C$ of the explorer neighbourhood $Ex_d(v_0,v_1)$ of length at most $d$. The cycles containing the edges $xv_x$ and $yv_y$ witness by \autoref{lem:3.4.} that $x$ and $y$ have unique copies in the explorer neighbourhood $Ex_d(v_0,v_1)$.
We denote these unique copies by $x$ and $y$, respectively. 
So $P$ is an $x$-$y$-path of $Ex_d(v_0,v_1)$. And as it avoids $v_0$ and $v_1$, the vertices $x$ and $y$ are in the same component of $Ex_d(v_0,v_1)-v_0-v_1$.
\end{proof}

\begin{rem}(Motivation)
We shall see that \autoref{lem:neigbInH} fairly easily gives one implication of \autoref{equi-def}; next we shall prove the other implication under the additional assumption that the vertices $v_0$ and $v_1$ have distance at least three and then we finish off by deducing the general case from this. 
\end{rem}

Our intermediate goal is to prove the following lemma.

\begin{setting}\label{X}
Fix a parameter $d\in \Nbb$, a graph\footnote{All graphs considered in this paper have neither parallel edges nor loops.} $G$, vertices $v_0$ and $v_1$ of $G$ of distance at most $d/2$ and distance at least three. 
Abbreviate the edge space $\mathcal{E}(C)$ of the connectivity graph $C=C_d(v_0,v_1,G)$ by $\mathcal{E}$.
\end{setting}

For a vertex $x$ of $Ex_d(v_0,v_1)$, we denote by $x'$ the unique vertex of $G$ of which $x$ is a copy. 

\begin{lem}\label{intermediate}
Assume \autoref{X}.
    If $x$ and $y$ are neighbours of $\{v_0,v_1\}$ in the same component of $Ex_d(v_0,v_1)-v_0-v_1$, then the vertices $x'$ and $y'$ are in the same component of $C$. 
\end{lem}

\begin{dfn}
Assume \autoref{X}.
We define a function $\varphi$ mapping cycles $o$ of $G$ of length at most $d$ to $\mathcal{E}$.
\begin{itemize}
\item If $o$ contains neither $v_0$ nor $v_1$, then we map $o$ onto the zero-vector; that is, $\varphi(o)=\vec{0}$.
    \item If $o$ contains exactly one of the vertices $v_i$, it contains exactly two edges $e$ and $f$ with end vertices in the set $\{v_0,v_1\}$.
     Let $x$ and $y$ be the end vertices of these edges different from either $v_j$. Then $o$ witnesses that $xy$ is an edge of the connectivity graph $C$, and we map $o$ to the characteristic vector of the edge $xy$. 
     \item Otherwise, $o$ contains both vertices $v_0$ and $v_1$. Thus it is composed of two $v_0$-$v_1$-paths, each having at least four vertices by \autoref{X}; call them $P^1$ and $P^2$. Let $x^j$ be the second vertex of $P^j$ and $y^j$ be its second but last vertex. 
     Since $P^j$ contains at least four vertices, the cycle $o$ witnesses that $x^jy^j$ is an edge of the connectivity graph $C$. We let $\varphi(o)$ be the characteristic vector of the set 
     $\{x^1y^1, x^2y^2\}$.
\end{itemize}
\end{dfn}

Given  a family $\mathcal{O}$ of cycles of $G$ of length at most $d$, we define $\varphi(\mathcal{O})$ via linear extension; that is:
$\varphi(\mathcal{O})=\Delta_{o\in \mathcal{O}} \varphi(O)$ (recall that the symmetric difference $\Delta$ is the addition over the binary field). 

For each $x\in N$, one of the pairs  $v_0x$ and $v_1x$ is an edge of $G$; and this pair is unique in the context of \autoref{X}.

\begin{lem}\label{pre-alg}
Assume \autoref{X}.
    Let  $\mathcal{O}$ be a family of cycles of $G$ of length at most $d$. Let $x$ be a vertex of $C$. 
    Then the parity\footnote{The \emph{parity} of a number is that number modulo two.} of the degree of $x$ in $\varphi(\mathcal{O})$ is equal to the number of cycles of $\mathcal{O}$ that contain one of the edges $v_0x$ or $v_1x$ (modulo 2).
\end{lem}
\begin{proof}
    If $\mathcal{O}$ consists of a single cycle, this follows from the definition of $\varphi$. And for general families $\mathcal{O}=(o_i|i\in I)$, it follows by linearity; that is: 
    $
    \varphi(\mathcal{O})= \Delta_{i\in I} \varphi(o_i)$.
\end{proof}

\begin{cor}\label{alg_fact}
Assume \autoref{X}.
    Let $o$ be an edge set of $G$ that is generated by a family $\mathcal{O}$  of cycles of $G$ of length at most $d$.
    Assume that $o$ has degree zero at $v_1$ but has degree two at the vertex $v_0$, and let $x$ and $y$ vertices such that $v_0x,v_0y\in o$.
    Then the support of $\varphi(\mathcal{O})$ is an edge set of the connectivity graph $C$ that has even degree at every vertex except for $x$ and $y$, where it has odd degree.
\end{cor}

\begin{proof}
By \autoref{X}, the pair $v_0v_1$ is not an edge and no vertex is a neighbour of both $v_0$ and $v_1$.
Now apply \autoref{pre-alg} to $\mathcal O$:
the edge set $o$ contains no edge of the form $v_1w$ for $w \in N(v_1)$. So vertices $w\in N(v_1)$ have even degree in the edge set $\varphi(\Ocal)$. Similarly, each $u \in N(v_0) - x-y$ has even degree in $\varphi(\Ocal)$. 
Finally, $x$ and $y$ have odd degree in $\varphi(\Ocal)$.
\end{proof}

Recall that 
for a vertex $x$ of $Ex_d(v_0,v_1)$, we denote by $x'$ the unique vertex of $G$ of which $x$ is a copy. 

\begin{lem}\label{cor1}\label{lem:SameCompSameV}
Assume \autoref{X}.
If $x$ and $y$ are neighbours of $v_0$ in the same component of the punctured explorer-neighbourhood $Ex_d(v_0,v_1)-v_0-v_1$, then $x'$ and $y'$ are in the same component of the connectivity graph $C$. 
\end{lem}
\begin{proof}
Since $x,y$ are in the same component, there must be an $x$-$y$-path in $Ex_d(v_0,v_1)-v_0-v_1$. Let $o$ be the cycle of $Ex_d(v_0,v_1)$ obtained by adding the edges $xv_0$ and $yv_0$ to the $x,y$-path at either end. By \autoref{lem:3.6.} and \autoref{lem:3.5.} the cycle $o$ is generated by a family $\mathcal{O}$ of cycles of $G$ of length at most $d$.

Denote by $F$ the support of $\varphi(\mathcal{O})$. 
Let $K$ be the subgraph of $C$ induced\footnote{The graph \emph{induced} by an edge set $F$ is the graph whose vertex set consists of the end vertices of $F$ and whose edge set is $F$.} by $F$.
Denote by $o'$ the symmetric difference over the family $\Ocal$ of cycles of $G$; it has degree zero at $v_1$ and degree two at $v_0$ and $v_0x',v_0y'\in o'$.
By \autoref{alg_fact} applied to $o'$, every vertex in $K$ has even degree, except for $x'$ and $y'$ which have odd degree.
By the handshaking lemma, the component of $K$ containing $x'$ contains $y'$. Thus $x'$ and $y'$ are in the same component of $C$.
\end{proof}

\begin{lem} \label{lem:SameCompDiffV}
Assume \autoref{X}.
Let $x$ be a neighbour of $v_0$ and $y$ be a neighbour of $v_1$ in the same component $K$ of the punctured explorer-neighbourhood $Ex_d(v_0,v_1)-v_0-v_1$. Then $x'$ and $y'$ are in the same component of the connectivity graph $C$. 
\end{lem}
\begin{proof}
\begin{sublem}\label{cl1}
There are vertices $\hat{x}, \hat{y}\in K$ adjacent to $v_0$ and $v_1$, respectively, so that $\hat{x}'$ and $\hat{y}'$ are adjacent in $C$. 
\end{sublem}
\begin{proof}
{\bf Case 1:} $Ex_d(v_0,v_1)-v_0-v_1$ is connected. Take a shortest $v_0$-$v_1$-path $P$. 
By \autoref{X}, $P$ has length at least three and at most $d/2$; in particular $P\se D_r(v_i)$ for $i=0,1$. Let $\hat{x}$ be its second vertex and $\hat{y}$ be its second but last vertex, which are distinct as $P$ contains at least four vertices. Then $P$ witnesses that $\hat{x}'$ and $\hat{y}'$ are adjacent in $C$.

    {\bf Case 2:} $Ex_d(v_0,v_1)-v_0-v_1$ is not connected. By \autoref{lem:3.10.}, there is a cycle $o$ of length at most $d$ that contains $v_0$, $v_1$ and a vertex of $K$. Let $P$ be a $v_0$-$v_1$-subpath of $o$ that contains a vertex of $K$; note that  $P\se D_r(v_i)$ for $i=0,1$. Let $\hat{x}$ be the second vertex of $P$ and $\hat{y}$ be its second but last vertex. By \autoref{X}, these two vertices are distinct. By definition, $\hat{x}'$ and $\hat{y}'$ are adjacent in $C$. 
    \end{proof}

By \autoref{cor1}, $x'$ and $\hat{x}'$ are in the same component of $C$. Similarly $y'$ and $\hat{y}'$ are in the same component of $C$.
By \autoref{cl1} and since being in the same component is a transitive relation, $x'$ and $y'$ are in the same component of the connectivity graph $C$.     
\end{proof}

\begin{proof}[Proof of \autoref{intermediate}.]
    Assume \autoref{X}. Let $x$ and $y$ be neighbours of $\{v_0,v_1\}$ in the same component of $Ex_d(v_0,v_1)-v_0-v_1$.
    By \autoref{cor1} or \autoref{lem:SameCompDiffV}, respectively, we conclude that $x'$, $y'$ are in the same component of $\mathcal C_d(v_0,v_1)$.
    \end{proof}

Let $G^k$ be the graph obtained from $G$ by subdividing each edge $k$ times.
We obtain the graph $C_{k\cdot d}(v_0,v_1,G^k)'$ from $C_{k\cdot d}(v_0,v_1,G^k)$ by deleting the subdivision vertices of the edge $v_0v_1$ if existent.

\begin{lem}\label{tech1}
Assume $d\geq 2$. 
    The graph $C=C_d(v_0,v_1,G)$ is connected if and only if   $C'=C_{k\cdot d}(v_0,v_1,G^k)'$ is connected.
\end{lem}

\begin{proof}
As $d\geq 2$, for every common neighbour $x$ of $v_0$ and $v_1$ the subdivision vertices of $v_0x$ and $v_0x$ in $C'$ are joined by an edge (as in $G^k$ there is a path of length exactly $2k$ linking them via $x$). 
Let $M$ be the set of all such edges; note that $M$ is a matching.
We obtain $C''$ from $C'$ by contracting the matching $M$. 

Now we define a map from $V(C)$ to $V(C'')$.
If $x\in V(C)$ is a neighbour of a single $v_i$, we map $x$ to the unique subdivision vertex of $xv_i$ in $V(C'')$.
Otherwise $x$ is in the common neighbourhood of $v_0$ and $v_1$, and we map $x$ to the contraction vertex of the corresponding edge of $M$.
It is straightforward to see that this map is bijective and induces a graph-isomorphism between $C$ and $C''$.
\end{proof}

\begin{lem}\label{tech2}
    $\{v_0,v_1\}$ is a $d$-local 2-separator of $G$ if and only if in the graph $G^k$, the punctured explorer-neighbourhood $Ex_{k\cdot d}(v_0,v_1)-v_0-v_1$ in $G$ has two components 
    that do not contain subdivision vertices of the edge $v_0v_1$.    
\end{lem}

\begin{proof}[Proof:]
    immediate from the definition of the explorer-neighbourhood.
\end{proof}

\begin{proof}[Proof of \autoref{equi-def}.]
Let $d\in \mathbb{N}$, and $G$ be a graph with vertices $v_0$ and $v_1$ of distance at most $d/2$.
Assume that the punctured explorer-neighbourhood $Ex_d(v_0,v_1)-v_0-v_1$ is connected.
Let $k=3$. And construct the graph $G^k$ as described above. Let $X$ be the set of subdivision vertices of the edge $v_0v_1$ if existent, otherwise $X=\emptyset$. 
Then by \autoref{tech2} $(Ex_{k\cdot d}(v_0,v_1)-v_0-v_1)\setminus X$ is connected. So by \autoref{intermediate}, $C'=C_{k\cdot d}(v_0,v_1,G^k)'$ is connected. So by \autoref{tech1},  $C=C_d(v_0,v_1,G)$ is connected.
To summarise, we have shown that if  $Ex_d(v_0,v_1)-v_0-v_1$ is connected, then $C$ is connected.

Conversely, assume that $C$ is connected. 
Let $x$ and $y$ be two arbitrary neighbours of $\{v_0,v_1\}$ in the explorer-neighbourhood $Ex_d(v_0,v_1)$. 
Then $x'$ and $y'$ are in the same component of $C$, so let $P$ be a path joining them. 
Due to \autoref{lem:neigbInH}, for every two consecutive vertices $a'$ and $b'$ on $P$ the vertices $a'$ and $b'$ have unique copies in $Ex_d(v_0,v_1)$ and these unique copies are in the same component of the punctured explorer neighbourhood $Ex_d(v_0,v_1)-v_0-v_1$.  Note that the unique copy of $Ex_d(v_0,v_1)$
for the starting vertex $x'$ of $P$ is $x$ and the unique copy for the end vertex $y'$ of $P$ is $y$. 
Since being in the same component is a transitive relation, $x$ and $y$ are in the same
component of the punctured explorer neighbourhood $Ex_d(v_0,v_1)-v_0-v_1$. 
We have shown that any two neighbours $x$ and $y$ of $\{v_0,v_1\}$ in the explorer-neighbourhood $Ex_d(v_0,v_1)$ are in the same component of the punctured explorer-neighbourhood $Ex_d(v_0,v_1)-v_0-v_1$; thus it is connected. 
\end{proof}

\begin{rem}
It is natural to consider the following strengthening of \autoref{equi-def}. Let $x$ and $y$ in $N(\{v_0,v_1\})$ in $Ex_d(v_0,v_1)$. Then $x$ and $y$ are in the same component of $Ex_d(v_0,v_1)-v_0-v_1$ if and only if $x'$ and $y'$ are in the same component of $C_d(v_0,v_1)$. This fact follows easily from \autoref{equi-def}. To see this add to the graph $G$ an edge from $x'$ to some vertex in every component of  $C_d(v_0,v_1)$ that does not contain $y'$ and apply \autoref{equi-def} in this new graph.
\end{rem}

\begin{rem}
    The results of this section extend to graphs with rational weights by considering corresponding subdivisions and re-scaling $d$ to a natural number.
    The generalisation to real-weights follows from the extension to rational weights as the graphs of this paper are finite.
\end{rem}

\section{The theory of local separators}\label{sec5}

\begin{rem}(Motivation)
    In connectivity theory, the most natural way to decompose a connected graph is to cut it along the articulation points into its two-connected components and its bridges (that is, single edges whose removal disconnects the graph). The interesting pieces of this decomposition are the two-connected components, which in a further step can be decomposed via Tutte's 2-separator theorem into the maximal 3-connected torsos and cycles. In this paper we decompose graphs in a similar way but instead of a global notion of separability, we work with the finer notion of local separators introduced above. Such graph decompositions -- for both local and global connectivity -- come with a decomposition graph that displays how the decomposition pieces are stuck together. In the case of global connectivity the decomposition graphs are always trees. In the finer context of local connectivity they can be genuine graphs. 
In this short section we summarise the basis for graph decompositions. 
\end{rem}

\begin{dfn}[{\cite[Definition 9.3]{carmesin2022local}}]
A \emph{graph decomposition} consists of a bipartite graph $(B, S)$ with bipartition classes $B$ and $S$, where the elements of $B$ are referred to as \enquote{bags-nodes} and the elements of $S$ are referred to as \enquote{separating-nodes}. This bipartite graph is referred to as the \enquote{decomposition graph}. For each node $x$ of the decomposition graph, there is a graph $G_x$ associated to $x$. Moreover for every edge $e$ of the decomposition graph from a separating-node $s$ to a bag-node $b$, there is a map $\iota_e$ that maps the associated graph $G_s$ to a subgraph of the associated graph $G_b$. We refer to $G_s$ with $s \in S$ as a \emph{local separator} and to $G_b$ with $b \in B$ as a \emph{bag}. 

The \emph{underlying graph} of a graph decomposition $(G_x \mid x \in V (B, S))$ is constructed from the disjoint union of the bags $G_b$ with $b \in B$ by identifying along all the families given by the copies of the graphs $G_s$ for $s \in S$. Formally, for each separating-node $s \in S$, its family is $(\iota_e(G_s))$, where the index ranges over the edges of $(B, S)$ incident with $s$.
\end{dfn}

We say that a $d$-local 2-separator $X$ \emph{crosses} a $d$-local 2-separator $Y$ if there is a cycle $o$ of length at most $d$ that contains $X$ and $Y$ and the cycle $o$ alternates between the sets $X$ and $Y$; that is, the set $X$ separates the set $Y$ when restricted to the subgraph $o$. 
As explained formally in \cite{carmesin2022local} and informally in  \autoref{fig:graph_deco}, the set of local separators of a graph decomposition does not contain a pair of crossing local separators. Conversely, any set of local separators without a crossing pair can be turned into a graph decomposition. 
A $d$-local 2-separator is \emph{totally nested} if it is not crossed by any $d$-local 2-separator. Like in Tutte's 2-separator theorem, it is most naturally to decompose a $d$-locally 2-connected graph along its totally nested 2-separators. As shown in \cite{carmesin2022local}, the decomposition torsos are $d$-locally 3-connected or cycles of length at most $d$.

\begin{figure}[htbp]
    \centering
        \includegraphics[width=\textwidth]{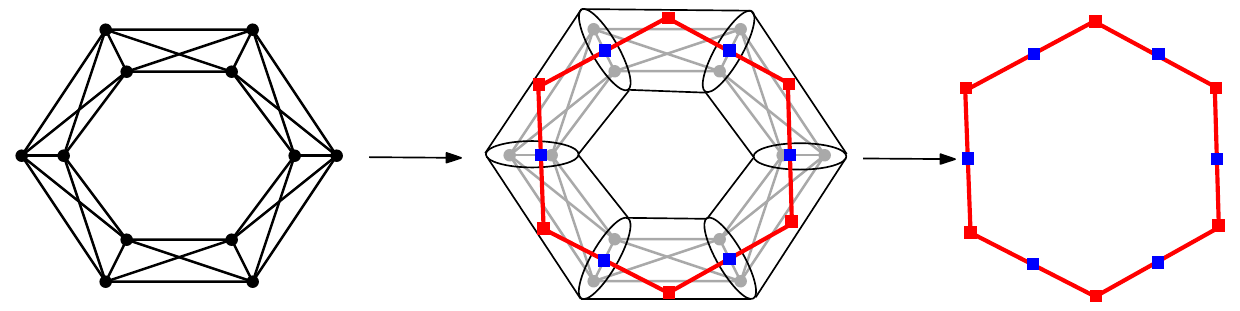} 
    \caption{Applying \autoref{thm:graphdecomp2} to the graph on the left, yields that this graph admits a graph decomposition whose decomposition graph is a cycle  (highlighted in red) and whose bags are $K_4$-s; that is, this graph can be obtained from $K_4$-s by gluing them together along a 6-cycle. The local separators that belong to this graph decomposition are the two vertices in the intersection of two bags whose nodes are adjacent in the red 6-cycle; the overall decomposition cycle is obtained from the red 6-cycle by subdividing each red edge with a blue vertex corresponding to the local separators. These local separators do not cross since they locally do not separate each other, compare the definition of crossing above. 
    }
        \label{fig:graph_deco}
\end{figure}

In this paper we provide an algorithm to compute for every parameter $d$ and every graph $G$, the $d$-local cutvertices of $G$. These local cutvertices directly give rise to a graph decomposition of $G$ as follows. 

\begin{thm}[{\cite[Theorem 4.1.]{carmesin2022local}}]\label{thm:graphdecomp1}
Given $d \in \mathbb N \cup \{ \infty \}$, every connected graph has a graph decomposition of adhesion one and locality $d$ such that all its bags are $d$-locally $2$-connected or single edges. This graph decomposition can be computed from the set of $d$-local cutvertices. 
\end{thm}

We provide an algorithm to compute for every parameter $d$ and every $d$-locally 2-connected graph $G$, the totally nested $d$-local 2-separators of $G$. These local 2-separators directly give rise to a graph decomposition of $G$ as follows. 

\begin{thm}[{\cite[Theorem 1.2.]{carmesin2022local}}]\label{thm:graphdecomp2} %
For every $d \in \mathbb N \cup \{\infty \}$, every connected $d$-locally $2$-connected graph $G$ has a graph decomposition of adhesion two and locality $d$ such that all its torsos are $d$-locally $3$-connected or cycles of length at most $d$. This graph decomposition can be computed from the totally nested  $d$-local $2$-separators of $G$.
\end{thm}

\autoref{thm:graphdecomp1} and  \autoref{thm:graphdecomp2} together form the first two steps of a local-global decomposition of every graph $G$.
Indeed, as with the block cutvertex theorem and Tutte's theorem, one applies \autoref{thm:graphdecomp2} to every $d$-locally 2-connected bag of the graph decomposition from \autoref{thm:graphdecomp1}. For $d=\infty$ local connectivity specialises to global connectivity.

\section{Algorithms}\label{sec:algorithms}

In this section we present algorithms to identify the local $1$-separators and local $2$-separators of an unweighted graph. Both algorithms follow the same structure: for every potential local 1-separator or local 2-separator we calculate a graph: the ball of diameter $d$ or the connectivity graph, respectively. A simple connectivity test on this graph determines whether we have a local 1-separator or local 2-separator, respectively. 
In the following, let $R$ be maximum size of a ball of diameter $d$ around a vertex in $G$, where the \emph{size} of a graph is the sum of its vertex-number and its edge-number.

\subsection{Finding $d$-local $1$-separators}\label{sec:alg1separators}

Per \autoref{sec:original-definitions} the ball $D_d(v)$ is the subgraph of $G$ consisting of those vertices and edges of $G$ on closed walks of length at most $d$ containing $v$. We can calculate both the vertex set and edge set of $D_d(v)$ with one breadth-first-search (BFS). A vertex $v \in G$ is called a $d$-local cutvertex if it separates the ball of diameter $d$ around $v$; formally if $D_d(v)- v$ is disconnected. 
The following algorithm is a direct translation of the definition of $d$-local cut-vertices into pseudo-code.

\begin{algorithm}
\caption{Determines the local $1$-separators for a given graph and diameter}
\begin{algorithmic}[1]
\REQUIRE Graph $G$, diameter $d$
\ENSURE Set of vertices that are $d$-local cut-vertices of $G$
\FUNCTION {find1Separators}{$G,d$}
\STATE $S \gets \emptyset$ 
\FORALL {$v \in V(G)$}
\STATE $D \gets $ calculate $D_d(v)$ using a BFS
\IF {$D-v$ is not connected}
\STATE $S \gets S \cup \{v\}$ 
\ENDIF 
\ENDFOR
\RETURN $S$
\ENDFUNCTION
\end{algorithmic}
\end{algorithm}

\begin{lem}\label{lem:time1-sep}
The algorithm for $\textsc{find1Separators}(G,r)$ has time complexity of $O(Rn)$.
On $n$ processors, it has parallel running time $O(R)$.
\end{lem}
\begin{proof}
The algorithm tests every vertex independently whether it is a $d$-local cutvertex, thus the for-loop in lines $3$-$8$ repeats $n$ times.
By capping the BFS in line $4$ at diameter $d$, it has time complexity $O(R)$. 
Analysing the connectivity of $D-v$ in line $5$ can be done using another breadth-first-search in $O(R)$.
Thus, the time complexity of one iteration of the for-loop is $O(R)$.
Thus, the overall time complexity of the algorithm is $O(Rn)$.

    Since every vertex is tested independently if it is a $d$-local cutvertex, the tests can be done in parallel. Thus, a parallel algorithm runs in $O(R)$ on $n$ processors.
\end{proof}

\subsection{Finding  $d$-local $2$-separators}\label{sec:2-seps}

A \emph{simplified connectivity graph} is obtained from the connectivity graph by deleting edges so that the vertex sets that form components are the same.
To decide whether a pair of vertices is a $d$-local $2$-separator, we calculate a simplified connectivity graph for these vertices. 
\begin{rem}(Motivation)
Since simplified connectivity graphs may have less edges than the connectivity graph itself, algorithms on them tend to run quicker. However, working with simplified connectivity graphs also allows us to detect local separators in the same way by \autoref{equi-def}. Hence working with simplified connectivity graphs is a minor technical detail to improve running times. 
\end{rem}

\begin{algorithm}\label{alg:connecivity}
\caption{Calculate a simplified connectivity graph of the vertices $v_0$ and $v_1$ and diameter $d$}
\begin{algorithmic}[1]
\REQUIRE Graph $G$, vertices $v_0, v_1$, diameter $d$
\ENSURE simplified $C_{d}(v_0, v_1)$
\FUNCTION {generateC}{$G,v_0,v_1,d$}
\STATE $C \gets$ empty graph
\STATE $V(C) \gets \{x \in V(G) \mid x$ is adjacent to $v_0$ or $v_1 \}$
\FOR{$i \in \{0,1\}$}
    \STATE $D_{i} \gets $calculate $D_d(v_i)$ using a BFS
    \STATE $D_{i} \gets D_{i} \setminus \{v_0, v_1\}$
    \STATE calculate the connected components of $D_i$
    \FORALL {component $S$ in $D_i$}
        \STATE $S_C \gets S\cap V(C)$
        \STATE $y \gets$ a vertex of $S_C$
        \FORALL {vertex $z \in S_C\setminus \{y\}$}
            \STATE add edge $yz$ to $C$
        \ENDFOR
    \ENDFOR
\ENDFOR

\RETURN $C$
\ENDFUNCTION
\end{algorithmic}
\end{algorithm}

\begin{lem}\label{lem:timeC}
The algorithm for $\textsc{generateC}(G,v_0,v_1,d)$ has time complexity of $O(R)$.
\end{lem}
\begin{proof}
The vertices $v_0$ and $v_1$ have each at most $R$ incident edges, thus calculating the vertex set of $C$ is in $O(R)$. 
The outer for loop from line $4$ to $13$ repeats twice. 
Calculating the balls  $D_i$ of diameter $d$ is in $O(R)$, as is calculating the connected components of $D_i$ in line $7$. 
Since the components $S$ of $D_i$ are disjoint, the two nested for-loops in line $8-13$ touch every vertex in $V(C)$ at most once. So adding these edges to $C$ can also be done in $O(R)$.
Combining these, the algorithm has time complexity $O(R)$.
\end{proof}

The \emph{cycle data} at a $d$-local separator $\{v_0,v_1\}$ is either the information that $v_0v_1$ is an edge or the local separator $\{v_0,v_1\}$ has at least three local components or else a cycle $o$ of length at most $d$ that contains both vertices $v_0$ and $v_1$ and interior vertices of two local components of $\{v_0,v_1\}$.
\autoref{alg:cycledata} computes the cycle data for a given $2$-separator.
In the first two cases $\{v_0,v_1\}$ is totally nested right away by the definition of crossing and in final case we check in time $O(d^2)$ whether a local separator of the third type is totally nested, as follows. 

\begin{algorithm}
\caption{Calculate the cycle data for a given local $2$-separator}\label{alg:cycledata}
\begin{algorithmic}[1]
\REQUIRE Graph $G$, vertices $v_0, v_1$, diameter $d$, simplified connectivity graph $C$
\ENSURE a cycle of length at most $d$
\FUNCTION {generateCycleData}{$G,v_0,v_1,d,C$}
    \STATE $c \gets$ number of components of $C\setminus \{v_0,v_1\}$
    \IF {$c \neq 2$ \OR $v_0v_1$ is edge of $G$}
    \RETURN $0$
    \ELSE 
		 \STATE $G \gets G-v_0$
        \STATE $left \gets$ shortest path from $v_1$ to any neighbour of $v_0$ in the first component of $C\setminus \{v_0,v_1\}$
        \STATE $right \gets$ shortest path from $v_1$ to any neighbour of $v_0$ in the second component of $C\setminus \{v_0,v_1\}$
        \RETURN cycle obtained by joining $left,right$ and $v_0$
    \ENDIF
\ENDFUNCTION
\end{algorithmic}
\end{algorithm}

\begin{lem}\label{cross-check}
Given a graph $G$ and the list $L$ of the $d$-local separators of $G$ with cycle data, we can check in time $O(d^2)$ whether a $d$-local separator of $G$ is totally nested. Computing the list $L'\se L$ of the totally nested $d$-local 2-separators takes time $O(d^2\cdot |L|)$. All these computations can be executed in parallel.
\end{lem}

\begin{proof}
 If a $d$-local separator $\{x,y\}$ crosses $\{v_0,v_1\}$ by a lemma from \cite{carmesin2022local}, the cycle $o$ alternates between $\{x,y\}$ and $\{v_0,v_1\}$. 
We now check whether any pair of vertices $(a,b)$ with $a$ on one of the subpaths of $o$ from $v_0$ to $v_1$ and $b$ on the other subpath if it is in the list $L$ of local separators.
Assuming constant time access to the list $L$, this takes time $O(d^2)$. 
So for all local separators of $L$ we need time $O(d^2\cdot |L|)$.
\end{proof}

\begin{algorithm}
\caption{Find local $2$-separators with cycle data}
\begin{algorithmic}[1]
\REQUIRE Graph $G$, diameter $d$
\ENSURE List of vertex pairs that are $d$-local $2$-separators of $G$
\FUNCTION{find2Separators}{$G,d$}
	\STATE $S \gets \emptyset$ 
	\FORALL {$v_0\in V(G)$}
        \STATE $P \gets \{v_1 \in V(G) \mid dist(v_0,v_1) \leq \frac{d}{2}\}$
        \FORALL{$w \in P$}
		      \STATE $C \gets$ \textsc{generateC}($G,v_0,v_1,d$)
            
		      \IF {$C\setminus \{v_0,v_1\}$ is not connected}
			     \STATE $S \gets S \cup \{(v_0,v_1)\}$ 
            \STATE $cycleData(v_0,v_1) \gets $\textsc{generateCycleData}($G,v_0,v_1,d,C$)
        \ENDIF
        \ENDFOR
\ENDFOR
\RETURN $S$
\ENDFUNCTION
\end{algorithmic}
\end{algorithm}

\begin{lem}\label{nr2}
Finding all $d$-local $2$-separators of a graph $G$ with cycle data can be done in time $O(R^2n)$.
On $n$ processors, it has parallel running time $O(R^2)$. The number of $d$-local $2$-separators is at most $Rn$.
\end{lem}
\begin{proof}

The outer for-loop (line $3$-$12$) repeats $n$ times.
Calculating the set $P$ in line $4$ can be done using a distance-capped BFS in $O(R)$. 
The set $P$ contains at most $R$ vertices. Thus, the inner for-loop (lines $5$-$11$) repeats at most $R$ times.
Due to \autoref{lem:timeC}, calculating the graph $C$ in line $6$ is in $O(R)$.
Since $C$ has at most $2R$ vertices, the connectivity test in line $7$ can be done in $O(R)$.  Calculating the cycle data of $(v_0,v_1)$ requires two breadth-first-searches. We cap the BFS at distance $d$, thus calculating the cycle data in line $9$ is in $O(R)$. Thus, one iteration of the inner for-loop has time complexity $O(R^2)$.
Combined, this implies an overall time complexity of $O( R^2n)$.

Since every vertex is tested independently if it is part of a $d$-local $2$-separator, the tests can be done in parallel. Thus, a parallel algorithm runs in $O(R^2)$ on $n$ processors. For a local 2-separator, we have $n$ choices for the first vertex and then at most $R$ choices for the second vertex, as it must lie within a ball of diameter $d$ around the first vertex. 
\end{proof}

\begin{lem}\label{nr2d}
Finding all totally nested $d$-local $2$-separators of a graph $G$ with can be done in time $O(max(d^2,R)\cdot R \cdot n)$.
\end{lem}

\begin{proof}
By \autoref{nr2}, we can compute the list of $d$-local 2-separators in time $O(R^2n)$, and this list has length at most $Rn$.
So by \autoref{cross-check} computing the sublist of the totally nested local 2-separators takes time $O(d^2\cdot R \cdot n)$.
This gives the desired running time. 
\end{proof}

\begin{eg}
When $G$ is a 2-dimensional grid, then $O(d^2)=O(R)$. So the maximum in the term of \autoref{nr2d} is equal to $R$ and in graphs with $R(d)$ at least as large as in the $2$-dimensional grid. So here we obtain an estimate of the running time by $O(R^2 \cdot n)$, whilst we have the general upper bound of $O(R^3 \cdot n)$.
\end{eg}

It would be most exciting if there was a \lq local analogue\rq\ of the linear time algorithm that computes the Tutte decomposition of 2-connected graphs, as follows.

\begin{oque}
Is there an algorithm that finds all totally nested $d$-local $2$-separators of an $d$-locally 2-connected graph $G$ with $n$ vertices in parallel running time in $O(R)$ on $n$ processors ?   
\end{oque}

\section{Finding structure in large graphs}\label{sec7}

In this section we demonstrate that the structure our algorithm detects in large graphs is indeed the \enquote{right} structure. 
For this we choose to analyse a road network, shown in \autoref{fig1}. Road maps do not come with a hard ground truth. But just looking at them we humans are able to identify some structure, like towns, rural areas and main connecting routes between them. 
We demonstrate that our algorithms are successful in identifying the structure, as a human perceives it.

While the data set of the road network comes with coordinates for the vertices,
it is important to note that our algorithm works on abstract graphs, not on geometric graphs. So the coordinates of the vertices are not part of the input. The reason we choose a geometric problem is that from the drawing many properties are visible with \enquote{bare eyes} but it is unclear how they can be computed. And we shall see that our algorithm is indeed able to compute some of them. So although the geometric information is not available to the algorithm, 
the outputs it produces are consistent with the geometry, and this gives us high confidence that it also will perform well on new examples where a ground truth is not known.

\subsection{Detecting structure in a road network}

We analyse the road network of North Rhine-Westphalia (NRW), see \autoref{fig:nrw_raw}. 
We choose this example, as NRW is a state with separate cities and one big urban-area, where several cities have \enquote{merged}, the Ruhr valley. Using both local $1$-separators and local $2$-separators we are able to identify the cities outside the Ruhr valley, the Ruhr valley and the separate cities in the Ruhr valley. The decomposition graph we obtain is a small structure graph, which retains the connectivity between the cities, see \autoref{fig:decomp}.

We extract the road network of NRW from OpenStreetMap (OSM) \cite{OpenStreetMap}.
This data set has 316000 nodes and 322000 edges. Based on the road network we compute the decomposition graph using local $1$-separators shown in \autoref{fig:decomp}. 
We can see that the algorithm is able to simplify the graph significantly by grouping vertices into locally $2$-connected clusters, while maintaining the characteristic structure. 
Quite often these cluster correspond to single cities (e.g. Cologne, Bonn or Paderborn). However, the algorithm groups the cities of the Ruhr valley in one big cluster. Hence, we take this cluster and analyse it further using local $2$-separators. The fine structure of the cluster is given by the cities of the Ruhr valley and the connections between them, as seen in the cutout of \autoref{fig:decomp}.
The algorithm shows us different levels of urbanisation and connectivity between the cities of NRW.

We construct the first decomposition graph for the local $1$-separators in a four-step process, as follows, see \autoref{fig:cutout_example} for an example. 
 
\begin{description}
    \item [Preprocessing.] The data set of the road network comes with positions of the vertices and no edge-weights. We use the positions to calculate the length of the edges. 
    The data set contains a lot of vertices of degree two that form long paths. Most of these vertices are local cut-vertices, however not interesting ones. Hence, we suppress vertices with degree $2$ as a preprocessing-step. We also delete vertices with degree one, since the other end of the unique edge is a global cut-vertices. The data set does not contain vertices of degree zero. Dealing with vertices of degree one and two beforehand is not strictly necessary, however it speeds up the computing time.
    The result of the preprocessing is a weighted graph with approximately $9000$ vertices and $14000$ edges.
\item[Compute $d$-local $1$-separators.] Using the algorithm detailed in \autoref{sec:alg1separators} we calculate the local cut vertices with $d=17$, as seen in \autoref{fig:exLocalCut}.

\item[Generate decomposition graph.] Using the local cutvertices, calculated in the previous step, we generate the block-cut graph. This is a bipartite graph where one partition class is the set of local cutvertices and the other partition class consists of one vertex for each cluster.

\item[Postprocessing.] To further simplify the graph, we suppress vertices of degree $2$. In the example of \autoref{fig:decomp} the corresponding graph has 224 vertices and 414 edges.
\end{description}

This concludes the decomposition with local $1$-separators. 
As an illustration of what is happening locally with the graph, we use a small cutout of the original graph shown in \autoref{fig:cutout_example}.
To further analyse the structure of any cluster, we repeat the process, this time with local 2-separators instead of local 1-separators.

\begin{figure}[htbp]
    \centering
     \begin{subfigure}[t]{0.23\textwidth}
     \centering        \includegraphics[scale=0.132]{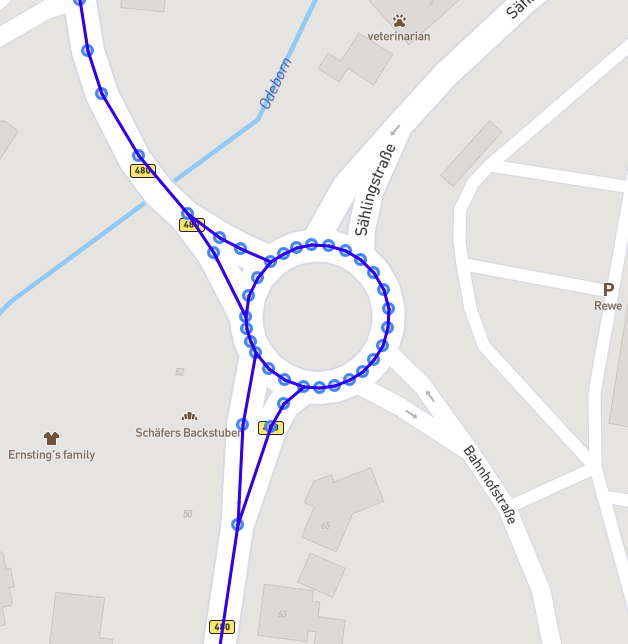} 
    \caption{Detail of the road network of NRW.}
\end{subfigure}
\hfill
    \centering
        \begin{subfigure}[t]{0.225\textwidth}
     \centering        \includegraphics[scale=0.237]{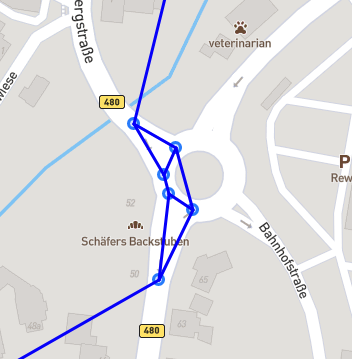} 
    \caption{After suppressing vertices of degree two.}
\end{subfigure}
\hfill
\begin{subfigure}[t]{0.23\textwidth}
     \centering        \includegraphics[scale=0.335]{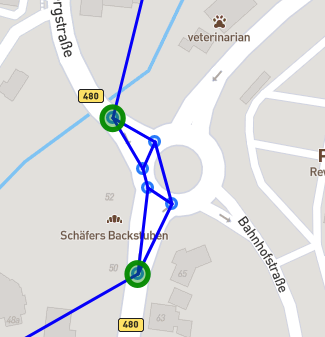} 
    \caption{Calculating local 1-separators (green).}
        \label{fig:exLocalCut}
\end{subfigure}
\hfill
\begin{subfigure}[t]{0.23\textwidth}
     \centering        \includegraphics[scale=0.237]{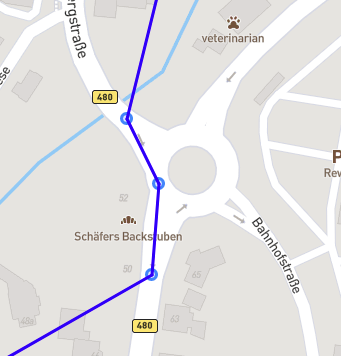} 
    \caption{Detail of block-cut graph.}
\end{subfigure}
        \caption{Detail of the road network and how it changes throughout the steps.}
        \label{fig:cutout_example}
\end{figure}

When analysing the road network using local $1$-separators, the cities of the Ruhr valley are grouped into one cluster. We now analyse this cluster using local $2$-separators. The process follows essentially the same steps as before. See \autoref{fig:cutout_example_2sep} for an example of what is happening locally with the graph.

\begin{description}
    \item [Preprocessing.] The clusters of the decomposition graph in \autoref{fig:decomp} are computed after the initial preprocessing step. So, the extracted cluster is a weighted graph and has no vertices of degree one and two. Thus, we can omit the preprocessing in this case.
\item[Compute $d$-local $2$-separators.] Using the algorithm detailed in \autoref{sec:2-seps} we calculate the local $2$-separators with $d=17$, as seen in \autoref{fig:ruhr_2sep}.
\item[Generate decomposition graph.] Using the local $2$-separators, calculated in the previous step, we generate the decomposition graph. This is a bipartite graph where one partition class is the set of local $2$-separators and the other partition class consists of one vertex for each cluster.

\item[Postprocessing.] We simplify the decomposition graph by suppressing vertices of degree two. 
\end{description}

This concludes the construction of the simplified decomposition graph via local $2$-separators. 

\begin{figure}[htbp]
    \centering
        \begin{subfigure}[t]{0.32\textwidth}
     \centering        \includegraphics[scale=0.30]{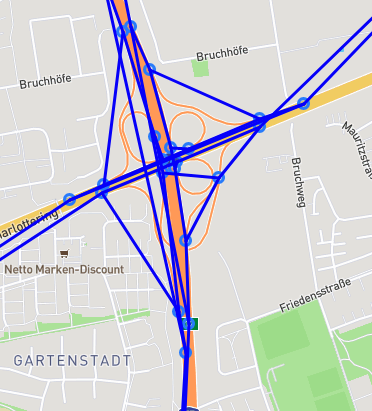} 
    \caption{Part of the road network of the Ruhr valley cluster showing a motorway junction.}
\end{subfigure}
\hfill
        \begin{subfigure}[t]{0.32\textwidth}
     \centering        \includegraphics[scale=0.30]{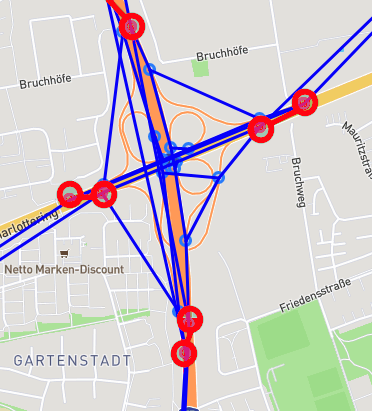} 
    \caption{Calculating local \mbox{$2$-separators} (red).}
        \label{fig:ruhr_2sep}
\end{subfigure}
\hfill
\begin{subfigure}[t]{0.32\textwidth}
     \centering        \includegraphics[scale=0.33]{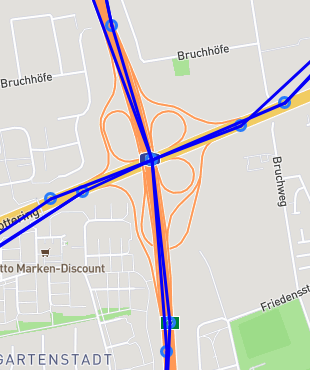} 
    \caption{The decomposition graph.}
\end{subfigure}
        \caption{Detail of the road network and how it changes throughout the steps.}
        \label{fig:cutout_example_2sep}
\end{figure}

\subsection{Different values for $d$}

Choosing different values for $d$ will result in a different number of local separators. 
In \autoref{fig:d} we can see that the number of local cutvertices of the NRW dataset decreases exponentially for increasing $d$.
Different values of $d$ will result in a different resolution of the detected structure of the data set. A smaller $d$ results in smaller clusters and thus a finer resolution. A larger $d$ will group more vertices together and the resolution will be coarser, resulting in a smaller decomposition graph. This is illustrated by the growing size of the biggest cluster with increasing $d$ in \autoref{fig:d}. The images in \autoref{fig:diffDecomps} show the different resolutions of the decomposition graph of the NRW road network in respect to the choice of $d$.

\begin{figure}[htbp]
    \centering
        \includegraphics[scale=0.45]{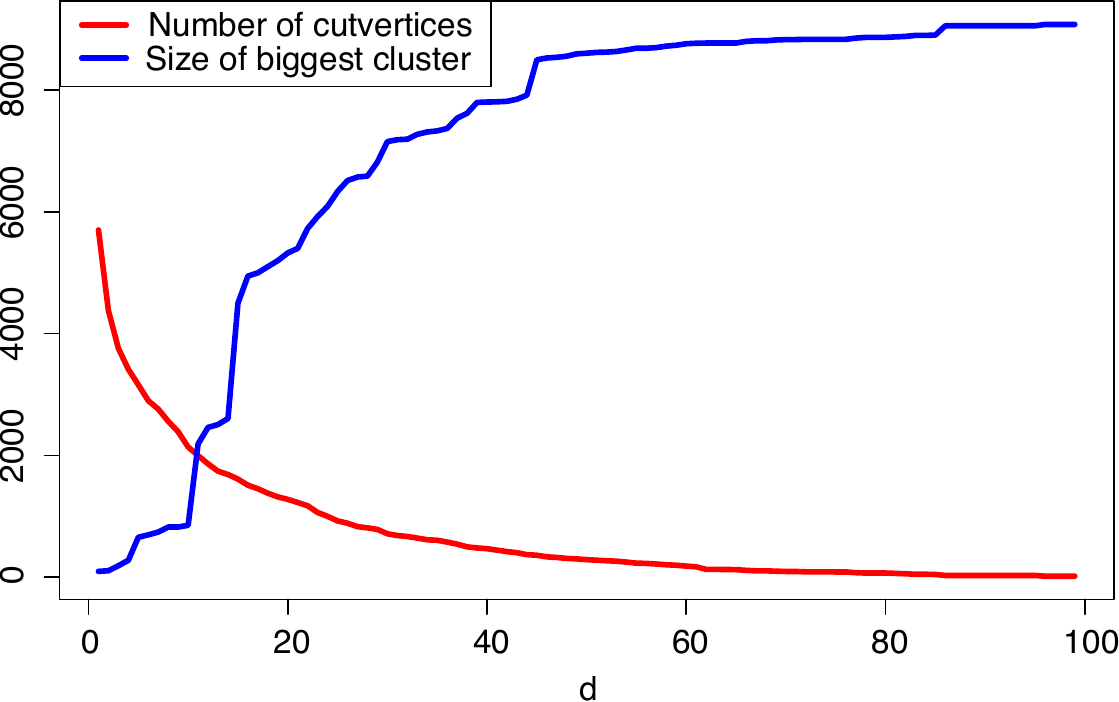} 
    \caption{The number of local cutvertices and the size of the biggest cluster in relation to the value of $d$ when analysing the road map of NRW.}
        \label{fig:d}
\end{figure}

\begin{figure}[htbp]
\begin{subfigure}{.5\textwidth}
\centering
\includegraphics[width=.95\textwidth]{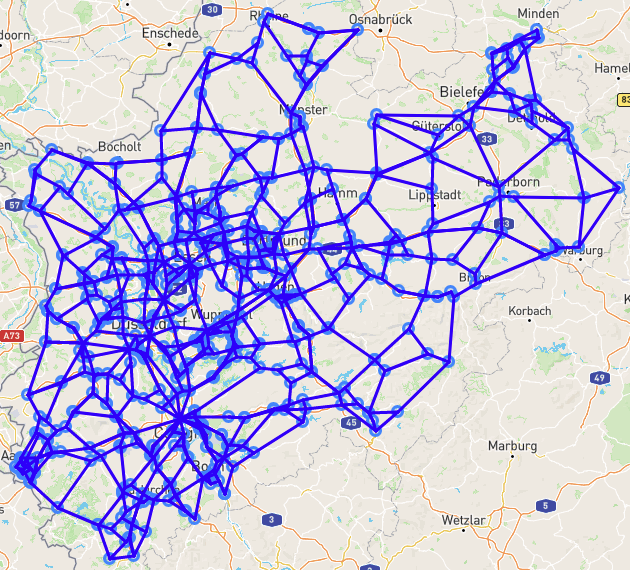}
\caption*{$d=5$}
\end{subfigure}
\hfill
\begin{subfigure}{.5\textwidth}
\centering
\includegraphics[width=.95\textwidth]{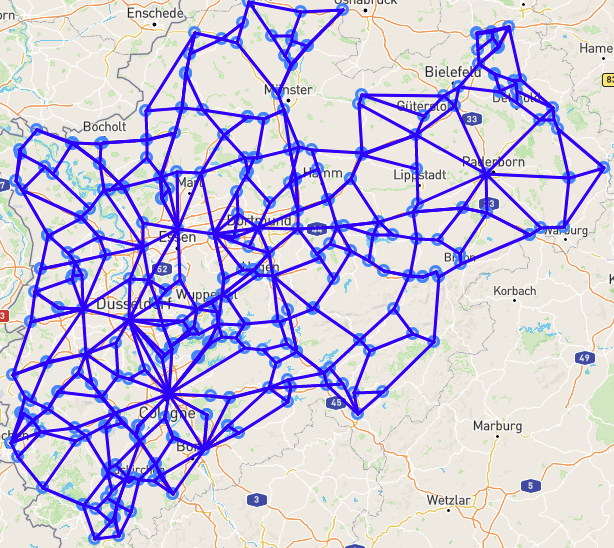}
\caption*{$d=11$}
\end{subfigure}
\hfill
\begin{subfigure}{.5\textwidth}
\centering
\includegraphics[width=.95\textwidth]{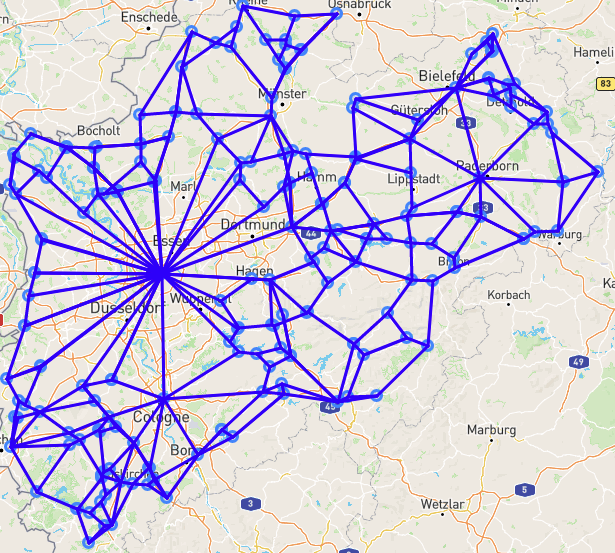}
\caption*{$d=17$}
\end{subfigure}
\hfill
\begin{subfigure}{.5\textwidth}
\centering
\includegraphics[width=.95\textwidth]{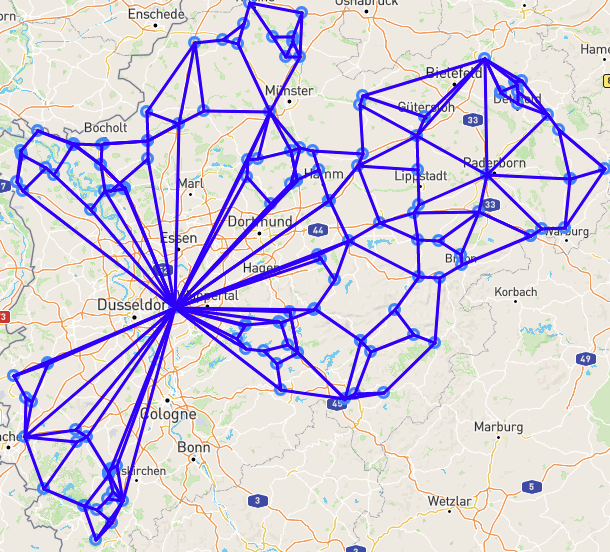}
\caption*{$d=23$}
\end{subfigure}
\caption{The detected structures in the road network of NRW for different $d$.}\label{fig:diffDecomps}
\end{figure}

\newpage

\section{Conclusion} \label{sec:conclusion}

Constructions of graph decompositions via local separators are a novel approach to compute local-global structure in large networks. We have implemented a fast algorithm to compute those, and have tested applicability on real world networks. It is exciting to apply this algorithm to further networks to investigate their clusters. 
In addition to this important and straightforward direction to continue our works, we also provide the following open problem in complexity theory. 

Whilst computing the independent set of a graph is an NP-hard problem~\cite{karp1972}, one direction to compute independent sets is to obtain algorithms for special classes of graphs. For example for a class of graphs of bounded tree-width the independent set can be computed in linear time ~\cite{courcelle1990monadic}. For graphs of large girth, local deletion algorithms in the sense of Hoppen, Lauer and Wormald \cite{{lauer2007large},{hoppen2018local}} can be used to prove asymptotic bounds on the independence number of large girth graphs in general and to approximate the independence number. A feature of these algorithms is that they can be executed in parallel at various vertices, reducing the overall running time. In \cite{bucic2023large} Buci\'{c} and Sudakov study the independence number of graphs provided some local bounds. This is vaguely related to the question we will be asking here. The difference is that we will impose a much stronger local assumption: that the local structure is given by a graph of bounded tree-width. 

Graph decompositions offer a framework that allows us to unify these two approaches, to find fast algorithms that approximate the independence number in graphs that have the global structure of a large girth graph and the local structure of a graph of bounded tree-width. Now we can formally describe this graph class: these are the graphs that admit graph decompositions of large locality and bounded width. As a first step towards this goal, we propose the following problem. Given parameters $r$, $d$ and $w$, denote by $\Gcal(r,d,w)$ the class of graphs that admit graph decompositions of locality $r$, width $w$ such that the decomposition graph is $d$-regular and has girth at least $r$ and additionally assume that all local separators of this graph decomposition are vertex-disjoint. 
We hope that the local deletion algorithm can be extended to the class $\Gcal(r,d,w)$, where the local rule is slightly more enhanced and makes use of the local structure of bounded tree width. 

\begin{oque}
Can you extend results of ~\cite{hoppen2018local} to the class $\Gcal(r,d,w)$, for some integer $r=r(d,w)$ which can be chosen arbitrary large in terms of $d$ and $w$? In particular, can you extend the estimates on the independence number from \cite{hoppen2018local} to $\Gcal(r,d,w)$?
\end{oque}

We expect that the answer to this question is affirmative and this phenomenon to hold for a large class of NP-hard problems. If true, then this means that the algorithms developed in this paper can be applied further to approximate the independence number in classes such as $\Gcal(r,d,w)$.

\bibliographystyle{plain}
\bibliography{literatur.bib}
\end{document}